\documentclass{article}\usepackage{amsthm}

\usepackage[utf8]{inputenc}
\usepackage{amsmath, amsfonts,xspace}

\newtheorem{mydef}{Definition}
\newtheorem{mylem}{Lemma}
\newtheorem{mypropo}{Proposition}
\newcommand{\F}{{\mathbb F}}
\newcommand{\Sf}{{\mathbb S}}
\newcommand{\SBox}{S-Box\xspace}
\newcommand{\SBoxes}{S-Boxes\xspace}

 \date{\today}
 \title{Generating S-Boxes from Semifields Pseudo-extensions}

\author{Jean-Guillaume Dumas \and Jean-Baptiste Orfila}
\date{\small Laboratoire J. Kuntzmann, Universit\'e de Grenoble. 51, rue des
  Math\'ematiques, umr CNRS 5224, bp 53X, F38041 Grenoble, France,
\href{mailto:Jean-Guillaume.Dumas@imag.fr,Jean-Baptiste.Orfila@imag.fr}{\{Jean-Guillaume.Dumas,Jean-Baptiste.Orfila\}@imag.fr}.
}

 \usepackage{color,svgcolor}
 \makeatletter
 \usepackage{hyperref}
 \hypersetup{
   pdftitle={\@title},
   pdfauthor={Jean-Guillaume Dumas and Jean-Baptiste Orfila},
   breaklinks=true, %
   colorlinks=true,
   linkcolor=blue,
   citecolor=darkred,
   urlcolor=darkblue,
 }
\makeatother

\begin{document}

\maketitle

\begin{abstract}
  Specific vectorial boolean functions, such as \SBoxes or APN functions have many
applications, for instance in symmetric ciphers. 
In cryptography they must satisfy some criteria (balancedness, high
nonlinearity, high algebraic degree, avalanche, or
transparency~\cite{Alvarez:2008:apn,Prouff:2005:transparency}) to provide best
possible resistance against attacks.
Functions satisfying most criteria are however difficult to find. 
Indeed, random generation does not
work~\cite{Danjean:2007:pasco,Hussain:2012:chaotic} and the \SBoxes used in the AES
or Camellia ciphers are actually variations around a single function, the inverse
function in $\F_{2^n}$. Would the latter function have an unforeseen weakness (for instance if more
practical algebraic attacks are developped), it would be desirable to have some
replacement candidates. For that matter, we propose to weaken a little bit the algebraic part of the
design of \SBoxes and use finite semifields instead of finite fields to build
such \SBoxes. Since it is not even known how
many semifields there are of order $2^8$, we propose
to build \SBoxes and APN functions via semifields pseudo-extensions of the form $\Sf^2_{2^4}$,
where $\Sf_{2^4}$ is any semifield of order $2^4$. Then, we mimic in this
structure the use of functions applied on a finite fields, such as the inverse or
the cube. We report here the construction of $12781$ non equivalent \SBoxes with with maximal nonlinearity, differential invariants, degrees
 and bit interdependency, and $2684$ APN functions.

\end{abstract}

\section{Introduction}

A substitution-box (denoted \SBox), is a tool used in
symmetric ciphers in order to increase their resistance against known
attacks, such as linear and differential cryptanalysis, by breaking cipher linearity.
\SBoxes are commonly represented by boolean functions i.e. $ S : \mathbb{F}_2^n
\to \mathbb{F}_2^m $ , whose dimensions $n,m$ depend on the cipher. 
For example, the AES \SBox uses $n=m=8$, views the finite field with $256$
elements as a vector space on its base field, and is generated by: 
\begin{equation}\label{eq:aes}\begin{array}{ccccc}
	T & : & \mathbb{F}_{2^8} & \to & \mathbb{F}_{2^8} \\
	& & 0 & \mapsto & 0 \\
	& & a & \mapsto & a^{-1} \\
\end{array}
\end{equation}

Once $T$ is computed, an affine transformation is
applied~\cite{Daemen:1998:rijndael}, and the result is in an excellent \SBox
from the point 
of view of security characteristics. More precisely, in the following, we will
use the list of criteria described in~\cite{Alvarez:2008:apn}. 
These criteria measure \SBoxes robustness with respect to possible attacks.
Among bijective \SBoxes, only AES and Camellia's \SBoxes have good scores on all
of these measures and both \SBoxes are build on a modified inverse computation.
Thus, would the latter function have an unforeseen weakness (for instance if
more practical algebraic attacks are developped), it would be desirable to have
some replacement candidates. 

Rather than trying different constructions, some works~\cite{Alvarez:2008:apn},
~\cite{Danjean:2007:pasco}  have been made on random searches among the $256!$
possibilities of bijective \SBoxes. 
Another approach is to design \SBoxes via the use of chaotic
maps~\cite{Hussain:2012:chaotic}.
Unfortunately, none of the \SBoxes build from these searches have the resistance
of AES against linear nor differential attacks.

Our idea is different, we replace the algebraic structure of AES and Camellia
(namely viewing the vector space as a finite field) by another structure, a
semifield. 
First, there exists different semifields of a given order up to isomorphism. 
Even when considering the more restrictive notion of isotopy~\cite{Albert:1960:finite}, the
semifields are still non unique. Therefore there could be several choices of
underlying structure, even with a single function.
Second, the nonzero elements of semifields still form a multiplicative group. 
Thus an inverse-like function could very well preserve good cryptographic
properties.
This idea could also be efficient for Almost Perfect non-Linear functions, defined as vectorial boolean function
verifying an almost perfect score to the differential invariant (see
also~\cite{Alvarez:2008:apn}, for instance). In the latter case, we will show that mimicking,
e.g., the cube function on semifields enables us to generate
APN functions.

In Section~\ref{sec:exten}, we recall the definition of semifields and propose a
construction of a degree~2 pseudo-extension of semifields of order $16$. 
From this construction we deduce bijective \SBoxes from $\F_2^8$ to $\F_2^8$, that mimic
the behavior of the above Function~(\ref{eq:aes}). 
We also present some efficient algorithms for semifields constructions in
Section~\ref{sec:gener}. 
Then we recall in Section~\ref{sec:crit}, the criteria that we use to rank the
obtained \SBoxes.
Finally, we show in Section~\ref{sec:expe} that our construction indeed yields
novel \SBoxes and APN functions that match the resistance of the best known ones.

\section{Semifields pseudo-extensions}\label{sec:exten}
In this section, after defining semifields, we describe the construction of
pseudo-extensions of degree $2$ of semifields containing $16$ elements. 

\begin{mydef}
A finite semifield $(\mathbb{S},+,\times)$ is a set $\mathbb{S}$ containing at least two elements, and associated with two binary laws (addition and multiplication), such that: 
\begin{enumerate}
	\item $(\mathbb{S}, +)$ is a group with neutral $0$
	\item $\forall a,b \in \mathbb{S}, ab=0 \Rightarrow a=0$ or $b=0$
	\item $\forall a,b,c \in \mathbb{S} : a(b+c)=ab+ac$ and $(a+b)c = ac+bc$
	\item $\forall a \in \mathbb{S}, \exists$ a neutral element for $\times$ denoted as $e$ which satisfies: $ea=ae=a$ \\
\end{enumerate}
\end{mydef}

Ideally, we would like to construct \SBoxes using:
$$\begin{array}{ccccc}
	T' & : & \mathbb{S}_{2^8} & \to & \mathbb{S}_{2^8} \\
	& & 0 & \mapsto & 0 \\
	& & a & \mapsto & a^{-1} \\
\end{array}$$

Unfortunately, we do not know the complete classification of these semifields for the moment. Currently, the 
largest classification in characteristic $2$ is of order 64~\cite{Rua:2009:classification_sf}.
Thus, in order to build \SBoxes with $256$ elements, we mimic the finite fields construction, based on a quotient structure: $\mathbb{F}_{2^8} = \mathbb{F}_{2^4}[X]/P_2$. 
However, the same notion of polynomial irreducibility is more difficult to
define in semifields, due to the possible non-associativity.

Actually, we just need to
build a bijection  $T': (\mathbb{S}_{2^4})^2 \to  (\mathbb{S}_{2^4})^2 $ as
close as possible to the inverse function, in order to (hopefully) take advantage of its cryptographic properties.
Therefore, we have to find an equivalent characterization to the polynomial irreducibility notion on in finite fields, applicable on semifields.
Let $P(X)=X^2 + \alpha X + \beta$, with $\alpha, \beta \in \mathbb{F}_{2^4} $, 
be an irreducible polynomial of degree $2$.
Elements of $\mathbb{F}_{2^8}$ viewed as $\mathbb{F}_{2^4}[X]/P$ are polynomials of degree $1$ of the form $aX+b$, denoted as couple $(a,b)\in \mathbb{F}^2_{2^4}$.  
Over the finite field $\mathbb{F}_{2^8}$, the inverse of $0X+b$ is $0X+u$, where $u=b^{-1}\in\mathbb{F}_{2^4}$ if $b\neq 0$. Then if $a\neq 0$, 
we let $\gamma \in  \mathbb{F}_{2^4}$ be such that $\gamma = a^{-1}b$, in order to obtain an unitary couple and thus simplify the following computations.
Then, the inverse of $aX+b$ is denoted $c'X+d'$ and we have $(aX+b)(c'X+d') = 1 \Leftrightarrow a(X+\gamma)(c'X+d') = 1$ or also $(X+\gamma)(cX+d)=1$, with $c'=a^{-1}c$ and $d'=a^{-1}d$. \\
After degree identification, and replacing $X^2=-\alpha X-\beta$, we obtain: 
\begin{equation}
\left \{
\begin{array}{ccc}
    d\gamma - c \beta & = & 1 \\
   c\gamma + d - c \alpha & = & 0 \\
\end{array}
\right. 
\end{equation}

Finally, we have:
\begin{equation}  
\left \{
\begin{array}{ccc}
    c & = & [(\alpha - \gamma)\gamma -\beta]^{-1} \\
  	d & = & c(\alpha - \gamma) \\
\end{array}
\right.
\label{eq2} 
\end{equation}

From the previous equations, it is now easy to deduce the following alternative characterisation of irreducible polynomials of degree $2$ over finite fields: 
\begin{mylem}\label{lem:charact}
Let $P : X^2 + \alpha X + \beta, \in \mathbb{F}_{2^4}[X] $, $P$ is irreducible if and only if $ \forall \gamma \in \mathbb{F}_{2^4}, [(\alpha - \gamma)\gamma -\beta]  \neq 0$.
\end{mylem}

Using Lemma~\ref{lem:charact}, we thus propose the following definition over semifields:
\begin{mydef}[Pseudo-irreducibility]\label{def:pseudoirred}
Let $P=X^2 + \alpha X + \beta \in \mathbb{S}_{2^4}[X] $, $P$ is pseudo-irreducible if and only if $ \forall \gamma \in \mathbb{S}_{2^4}, [(\alpha - \gamma)\gamma -\beta]  \neq 0$.
\end{mydef} 

Thus, in the case where $\mathbb{S}_{2^4} \simeq  \mathbb{F}_{2^4}$,
our pseudo-irreducibility notion reduces to irreducibility. 
Now we are able to define our pseudo-inversion as:
\begin{mylem}\label{lem:pseudoinv}
Let $P : X^2 + \alpha X + \beta, \in \mathbb{S}_{2^4}[X] $ be a pseudo-irreducible polynomial. The transformation: 

$$\begin{array}{ccccc}
	T' & : & (\mathbb{S}_{2^4})^2 & \to & (\mathbb{S}_{2^4})^2 \\
	& & (0,0) & \mapsto & (0,0) \\
	& & (0,b) & \mapsto & (0, b^{-1}) \\	
	& & (a,b) & \mapsto &(a^{-1}c,a^{-1}d)\\
\end{array}$$

such that $\gamma = a^{-1}b, c=[(\alpha - \gamma)\gamma -\beta]^{-1}$, and $d=c(\alpha - \gamma)$, is a bijection.
\end{mylem}

\begin{proof}

In the case where $a=0$, $T'$ is obviously one-to-one.
Let us assume now that $a \ne 0$. 

For proving injectivity, we suppose that $\exists \gamma_1, \gamma_2 \in \mathbb{S}_{2^4}$ such that $c(\alpha-\gamma_1)=c(\alpha-\gamma_2)$.
Then $c\alpha-c\gamma_1=c\alpha-c\gamma_2$, so that $c(\gamma_1-\gamma_2)=0$ and thus $c^{-1}c(\gamma_1-\gamma_2)=0$. Finally $\gamma_1 = \gamma_2$.

Then, as $\mathbb{S}^2_{2^4}$ has a finite cardinality, any injective endofunction 
is bijective.
\end{proof}

\section{Semifields efficient generation}\label{sec:gener}
As a prerequisite for constructing pseudo-extensions, we need semifields of order $2^{4}$. In this section, we expose some results and optimizations about efficient generation of semifields.
Recent results about semifields are detailed in~\cite{Combarro:2011:advances_sf} and in particular, they  show that we can represent semifields as matrix vector spaces:

\begin{mypropo}[~\cite{Combarro:2011:advances_sf}, Prop 3.]
There exists a finite semifield $\mathbb{S}$ of dimension $n$ over $ \mathbb{F}_q \subseteq \mathbb{S} $ iff there exists a set of $n$ matrices  $\{A_1,...A_n\} \subseteq GL(n,q)$ such that: 
\begin{itemize}
\item $A_1$ is the identity matrix
\item $\displaystyle\sum_{i=1}^n \lambda_iA_i \in GL(n,q), \forall (\lambda_1,...,\lambda_n) \in \mathbb{F}_q^n  \backslash \{0\} $ 
\item The first column of the matrix $A_i$ is the column vector with a $1$ in the $i^{th}$ position, and $0$ everywhere else.\\
\end{itemize}
\end{mypropo}

This proposition is fundamental, since it allows us to use efficients matrix computations to discover new semifields.
In our case, we restrict this proposition to $q=2$, and $n \le 8 $. 

In order to generate semifields, we use the algorithms described in~\cite{Rua:2009:classification_sf}. 
The idea is to select lists of matrices extracted from $GL(n,2)$ with a prescribed first column.
It is thus necessary to check invertibility of all possible linear combinations, in order to gradually reduce the possible semifield candidates.
In practice, the invertibility check is done by a determinant computation.
Then, in order to accelerate the process, some combinations of matrices can be discarded, as they can yield already found spaces. This is formalized via the notion of {\em isotopy} of semifields:
\begin{mydef}[Isotopy]\label{def:isotopy}
Let $\mathbb{S}_1$and $\mathbb{S}_2$ be two semifields over the same finite field $\mathbb{F}_p$, then an isotopy between $\mathbb{S}_1$ and $\mathbb{S}_2$ is a triple $(F,G,H)$ of bijective linear maps $\mathbb{S}_1 \to \mathbb{S}_2$ over $\mathbb{F}_p$ such that $H(ab) = F(a)G(b), \forall a,b \in \mathbb{S}_1$
\end{mydef}
Definition~\ref{def:isotopy} is used to define an equivalence relation between semifields, which can be verified with the help of matrix multiplications, see~\cite[Prop. 2]{Combarro:2011:advances_sf}.

Even if only square matrices with small size are involved, semifield generation remains complex for the large amount of computations involved. 
For instance, generation of all matrices constituting $GL(8,2)$ could require $2^{64}$ determinant computations. 

For $\Sf_{2^4}$ this is less important, but any improvement becomes substantial
if you are to generate many \SBoxes and experiment with them. We thus propose in
the following some optimization for the computation of the determinant and of
matrix multiplication, based on tabulation and Gray codes. 

\subsection{Optimizing determinant using double level Gray codes and tabulations}
Classical determinant computations use Gaussian elimination, with a $O(\frac{2}{3}n^3)$ complexity for a single determinant computation. 
Thus, in order to build $GL(n,2)$ by testing all the possible matrices, we obtain an overall complexity of $O(\frac{2}{3}n^3 2^{n^2})$. 
Here, we present two ways to reduce this complexity.

The first optimization is about tabulating the computations via the recurrence
formula of the Laplace expansion of determinants:
\begin{itemize}
	\item If $A$ is $1\times 1$ matrix, $det(A) = a$, with $A=(a)$.
	\item Otherwise, $n \ge 2$,  $$det(A) = \displaystyle \sum_{i=1}^n (-1)^{i+j}a_{j}M_{1j}$$ with $M_{1j}$ the determinant of the submatrix defined as A deprived of its first row and of its $j^{th}$ column (we chose the first row deletion and the column development arbitrarilly). 
\end{itemize}

More precisely, since we have to compute all determinants for each matrix size, the idea is to store them in order to accelerate the computations of the larger matrix dimension. 
By doing this, we replace a sub-determinant computation by a table access. The drawback of this method is the memory limitation, and we succeeded to apply it for square matrices up to size $6$. Indeed, for $n=7$, we should store $2^{49}$ computations, that being around 500 Tb of data.

Our second optimization is about improving the way of passing through all
matrices. Since each matrix has a unique integer representation (using the $n^2$
bits as digits), the easiest method to go through all the determinants is to
increment this integer representation until its largest value. However, it implies "random" modifications on the matrix binaries coefficients. By using a Gray code, which allows to pass from a value to another by modifying only one bit between them, we are thus able to pass from one determinant to the other by modifying only one term in the sum:
the idea is to cut the matrix in two parts, the first row on the one hand, containing $n$ bits, and the remainding ${n(n-1)}$ coefficients, which we call the base, on the other hand.
Then we apply two distinct Gray codes, one for each hand. 
First, we fix a value for the base, and then we go all over possible values for the first row, following a Gray code on this row. 
Second, we change the base value with another dedicated Gray code, and go again through all possible values from the first line. 
Memory exchange is thus reduced because we only need to access the lower dimension table $n$ times for each possible submatrix determinant, but for $2^n$ computations. 

The complexity is also drastically reduced by linking successive computations. Indeed, by modifying only one bit between two values, the determinant computation is reduced to the following formula: 
$\Delta_{k} = \Delta_{k-1} \oplus M_{1j}$, where $ M_{1j}$ is defined as in the previous formula, and $\Delta_0 = 0 $, since in a Gray code the first number is $0$. 
Thus, we reduced the determinant computation, which would normally requires $n-1$ XOR operations to only one, for $n \le 7$. 

Finally, we obtain the following lemma: 
\begin{mylem}
Let $n$ be the dimension of squares matrices, $n \le 6$, then the complexity of the determination of $GL(n,2)$, using tabulation and Gray codes, is bounded by $D_n$ that satisfies:
\begin{equation}
D_n = 2^{n^2} + O(2^{n^2-n-2})
\end{equation}
\end{mylem}

\begin{proof}
The complexity of the above algorithm is obtained by counting XOR operations and is given by the following recurrence formula:  
\begin{equation}
\left \{
\begin{array}{ccc}
    D_1 & = & 0 \\
  	D_n & = & D_{n-1} + 2^{n^2} - 1 \\
\end{array}
\right. 
\end{equation}
Therefore, we have:
$$\begin{array}{ccc}
	D_n  & = & \sum\limits_{i=2}^{n} \left(2^{i^2} -1\right)  \\
	&  = & \sum\limits_{j=0}^{n-2} 2^{(n-j)^2} - (n-1) \\
	& = & 2^{n^2}\sum\limits_{j=0}^{n-2} 2^{-2nj+j^2} - (n-1) \\
	& \le & 2^{n^2}\sum\limits_{j=0}^{n-2} 2^{-2nj+(n-2)j} - (n-1)\\
	& = & 2^{n^2}\sum\limits_{j=0}^{n-2} 2^{(-n-2)j} - (n-1) \\ 
	& = & \frac{1 - 2^{(-n-2)(n-1)}}{1 - 2^{-n-2}} -(n-1)\\
	& \le & 2^{n^2} (1 + \frac{1}{2^{n+2}-1}) - (n-1)
\end{array}$$
\end{proof}
We thus have a gain of a factor $\frac{2}{3}n^3$ over the naive Gaussian
elimination.

\subsection{Optimization of matrix multiplication by tabulation}

Similarly, we can optimize matrix multiplication with some tabulations. 
A first step consists in computing and storing all
possible products between all values of the first row of the left matrix, and
half the right one. Then, the matrix product $A \cdot B$ will simply  be
obtained by two table accesses per row of $A$ (one for each half of
$B$). Therefore, for computing $k$ products, 
we obtain a complexity bound of $O(2^{n+ n \lceil \frac{n}{2}
  \rceil}(n-1) + k2n) $. 
As comparison, if we used a scalars products and transposition algorithms, we
would have a complexity of $O(kn^2(n-1))$. As a conclusion, we see that our
optimization is interesting only if $k > \frac{2^{n+ n \lceil \frac{n}{2}
  \rceil}}{n^2 + \frac{2n}{n-1}}$. 
For semifields generation, the number of equivalence test is of the order of $2^{n^2}$ and the second term of the first complexity is therefore dominant.
In this case our optimization allows to gain another factor of $n^2$ in the complexity bound.

\section{\SBoxes criterion}\label{sec:crit}

Several criteria have been defined to measure \SBoxes resistance when faced to different types of attacks. In order to select our \SBoxes, we have chosen the following criteria, following mostly~\cite{Alvarez:2008:apn}. 
We denote by $S$ the \SBox function.
\begin{enumerate}
\item Bijectivity. By construction we only look for bijective functions.
\item Fixed Points. We favor functions without any fix points nor reverse fix points (as for the AES, this can be avoided by some affine transform on the trial). 
\item Non-linearity. We return the linear invariant $\lambda_S$, defined as follows:  
						$$\lambda_S = \displaystyle \max_{a,b\in \mathbb{F}_{2^n}, b\ne 0} \{ |-2^{n-1} + \#\{x \in \mathbb{F}_{2^n} : (a|x) \oplus (b|S(x)) = 0 \} | \} $$
\item XOR table and differential invariant.  A XOR table of S is based on the
  computation of $\delta_S(a,b) = \{x \in \mathbb{F}_{2^n} : S(x) \oplus S(a
  \oplus x) = b \}, \forall a,b\in \mathbb{F}_2^n$. The differential invariant
  $\delta_S$ is equal to  
$$ \delta_S=\max_{a,b\in \mathbb{F}_2^n, a\ne 0} |\delta_S(a,b)| 
=\max_{a,b\in \mathbb{F}_2^n, a\ne 0} | \{x \in \mathbb{F}_{2^n} : S(x) \oplus S(a
  \oplus x) = b \}|
.$$ 

\item Avalanche. Strict avalanche criterion of order $k$ ($SAC(k)$) requires
  that the function $x \mapsto S(x) \oplus S(a \oplus x)$ stays balanced for all
  $a \in \mathbb{F}_{2^n}$ of weight inferior to $k$. The goal is to provide a
  1/2 probability of outputs modifications in case of $k$ bits complemented for
  entries.  
 In our case, we measure the distance of the \SBox to $SAC(1)$, component
 function by component function, and we denote by $A_S = \displaystyle \max_{i=1..8} |
 2^{n-1}-\sum_x S_i(x) \oplus S_i(a \oplus x)|$ the obtained maximum.

\item Bit independance. Bit independancy is modelized by the computation of
  SAC(1) on the function defined by the sum of any two columns or any column of
  the matrix representation of S. As previously, we then measure its distance to
  SAC(1), and we denote it by $B_S$.  

\item Transparency. This notion has been introduced by Prouff in~\cite{Prouff:2005:transparency}, and allows to measure the resistance of \SBoxes against differential power analysis. 
The definition is the following: 
\begin{multline*}
  T_S = \max_{\beta \in \mathbb{F}_2^n} \biggl(\left|n - 2H(\beta)\right|\\ - 
    \frac{1}{2^n(2^n-1)}\sum\limits_{a \in \mathbb{F}_2^n} \Bigl| \sum\limits_{v
        \in \mathbb{F}_2^n, H(v)=1}(-1)^{v\beta}W_{D_a,S}(0,v)\Bigr|\biggr)
\end{multline*}
 where $W_{D_a,S}(u,v) = \sum\limits_{x \in \mathbb{F}_2^n}(-1)^{v[S(x) +
   S(x+a)] + ux} $ and $H(x)$ is the hamming weight.
\end{enumerate}

By using this criteria, we are able to compare the efficiency of our \SBoxes with the already existing ones.

For instance, the \SBoxes of AES and Camellia have minimal non-linearity $\lambda_{AES}=\lambda_{Camellia}=16$, minimal
differential invariant among non-APN functions,
$\delta_{AES}=\delta_{Camellia}=4$, very good bit independence
$B_{AES}=B_{Camellia}=8$ and avalanche criterion with Camellia slightly better
on the latter: $A_{AES}=8$ and $A_{Camellia}=6$.

\section{Experimental results}\label{sec:expe}
\subsection{Results of \SBoxes based on semifields
  pseudo-extensions}

We have implemented a simple matrix arithmetic using our optimizations, in order
to generate semifields of order 16 plus their pseudo-extensions. We have then
constructed \SBoxes with the help of the pseudo-inverse bijection of
Lemma~\ref{lem:pseudoinv}, and apply all the tests of Section~\ref{sec:crit}. 

We managed to generate 19336 semifields of order $2^{4}$ (with possible isomorphic ones). 
In average, $98$ polynomials per semifields
were pseudo-irreducible, with a minimum of $91$ and a maximum of $120$ (the
latter for the finite field).
By testing all possible pseudo-irreducible polynomials for each semifield, we
obtained $12781$ \SBoxes, with maximal nonlinearity, differential invariants,
degrees and bit interdependency. 

Among the latter $8364$ had fix points, and among the ones without fix points,
$4122$ had avalanche equal to $8$ (as good as AES) and $288$ had avalanche
equal to $6$ (as good as Camellia).
Among the ($4122+288$) latter \SBoxes, $863$ have a better transparency level than the inverse function on
a finite field. \\

\begin{table}[htb]\center
\begin{tabular}{|c|c|c|c|c|c|c|c|c|c|}
  \hline 
  & $\delta $ &  $\lambda $ & Alg deg & Poly deg & 
  Fix point & Av & Bi &  Tr \\
  \hline
  AES & 4 & 16 & 7 & 254 & 0 & 8 & 8  & 7.85319 \\
  Camellia1 & 4 & 16 & 7 & 254 & 0 & 6 & 8 & 7.85564 \\
  15306 & 4 & 16 & 7 & 254 & 0 & 8 & 8 & 7.84314 \\
  19203 & 4 & 16 & 7 & 254 & 0 & 6 & 8 & 7.84804 \\
  \hline
\end{tabular}
\caption{Some resistance criteria for \SBoxes}\label{tab:critbox}
\end{table}

For instance, in Table~\ref{tab:critbox}, $'15306'$ represents the \SBox obtained with an avalanche equals to $8$, and
lowest transparency score. The \SBox $'19203'$ has the best transparency score
with an avalanche equals to 6. \\

To illustrate our approach, we give here the construction of our \SBox
$'19203'$. It is generated by the semifield generated by the linear combinations
of the matrices in Equation~(\ref{eq:sf19203}) below, with $X^2 + 6X + 1$ as 
pseudo-irreducible polynomial for the pseudo-extension, with $6$
corresponding to $0110$ in binary and thus to the linear combination $A_2 + A_3$.
\begin{equation}\label{eq:sf19203}
\left \{ A_1 = 
\begin{pmatrix}
  1 & 0 & 0 & 0 \\
  0 & 1 & 0 & 0 \\
  0 & 0 & 1 & 0 \\
  0 & 0 & 0 & 1
\end{pmatrix} , 
A_2 =
\begin{pmatrix}
  0 & 0 & 0 & 1 \\
  1 & 1 & 0 & 0 \\
  0 & 1 & 0 & 1 \\
  0 & 0 & 1 & 0
\end{pmatrix} , 
A_3 =
\begin{pmatrix}
  0 & 0 & 1 & 1 \\
  0 & 1 & 0 & 0 \\
  1 & 1 & 1 & 1 \\
  0 & 1 & 0 & 0
\end{pmatrix}, 
A_4 =
\begin{pmatrix}
  0 & 1 & 0 & 1 \\
  0 & 1 & 1 & 1 \\
  0 & 1 & 1 & 0 \\
  1 & 0 & 0 & 1
\end{pmatrix} 
\right \} 
\end{equation}

Finally, we get Table~\ref{tab:23016} which shows the $'19203'$ \SBox in
hexadecimal.

\begin{table}[htb]\center
\begin{tabular}{|c||c|c|c|c|c|c|c|c|c|c|c|c|c|c|c|c|}
  \hline
  & 0 & 1 &2 &3 &4 &5&6&7&8&9&a&b&c&d&e&f \\
  \hline
  \hline
  0&  3f & 20 & 9a & f9 & 5c & 43 & d8 & a4 & bb & 7d & 1e & 85 & c7 & 62 & e6 & 1  \\ \hline
  1&  8c & b9 & 80 & 39 & a1 & 9c & ce & a6 & 2c & 97 & 5d & 9d & c6 & a3 & 4f & 6f \\ \hline
  2&  5b & aa & de & 61 & ab & 32 & 24 & 22 & 9e & 3d & 4c & ca & 7b & e5 & 65 & d6 \\ \hline
  3&  b4 & bf & 4b & 35 & fb & b6 & 6b & 50 & 53 & 5 & 92 & f3 & e4 & 4e & 29 & 33 \\ \hline
  4&  d0 & 40 & 4a & bc & d4 & 45 & 49 & 10 & e0 & b7 & 6c & 8f & c4 & 9 & 82 & 8 \\ \hline
  5&  63 & db & 7f & f1 & e3 & 52 & 13 & 2a & 28 & 60 & 5f & f8 & ec & eb & 2e & c2 \\ \hline
  6&  5e & 25 & 4 & 41 & 69 & 95 & 72 & 34 & 75 & 4d & 31 & ac & 26 & f0 & b2 & 83 \\ \hline
  7&  2 & a & 84 & 5a & 57 & 86 & ff & 1f & 30 & 14 & 36 & 88 & d2 & d7 & 70 & 74 \\ \hline
  8&  b1 & 6 & d3 & 98 & 87 & 8e & 38 & 77 & 99 & 96 & 8a & 67 & 46 & 6d & f5 & 1d \\ \hline
  9&  3a & 1b & 37 & ee & 3b & 81 & e1 & df & d1 & 93 & cc & 91 & b8 & 3c & 51 & a9 \\ \hline
  a&  d5 & 1a & 2b & 59 & b & 12 & bd & f7 & a0 & 2d & 78 & 76 & 71 & cd & 8b & 18 \\ \hline
  b&  e8 & 11 & ad & be & e2 & 7e & 0 & a8 & cb & 9b & fa & 58 & 9f & ef & f6 & 94 \\ \hline
  c&  ed & 27 & ba & f & 2f & d & c & 54 & 21 & 73 & b0 & 19 & f4 & 8d & c8 & 6e \\ \hline
  d&  89 & 48 & c5 & 23 & 64 & 47 & 7c & 16 & c1 & fd & e7 & cf & ea & 15 & da & a7 \\ \hline
  e&  7 & e9 & c3 & 44 & a2 & e & 79 & 7a & 3e & 90 & 6a & fc & a5 & 56 & b3 & dd \\ \hline
  f&  66 & c9 & dc & b5 & ae & af & 68 & f2 & 17 & 42 & 55 & d9 & 3 & c0 & 1c & fe \\ \hline

\end{tabular} 
\caption{An \SBox generated from a semifield with maximal linear and differential invariants}\label{tab:23016}
\end{table}

\subsection{APN functions based on semifields}
Vectorial boolean functions obtaining the best possible result for the $\delta$ invariant,
i.e. $\delta =2$, are called almost perfect non-linear functions (denoted APN). For instance,
in~\cite{Alvarez:2008:apn}, the cube
function is APN over the finite field with $256$ elements. As previously, we mimic this function on
semifields, instead of finite fields. In $\F_{2^8} =
\mathbb{F}_{2^4}[X]/P$, with $P$ an irreducible polynomial of the form
$X^2+\alpha X + \beta$, the cube function is defined as $(aX+b) \mapsto (cX+d)$ such that $(aX+b)^3 = (cX+d),
a,b,c,d \in \F_{2^8}$. One of the possibilities is: 
\begin{equation}
  \begin{array}{lclclc}
    c & = &[(aa^2)\alpha]\alpha - (aa^2)\beta +a(ab) +a(ba) + ab^2 -(ba^2)\alpha \\
    d & = &[(aa^2)\alpha]\beta -(ba^2)\beta +b(ab) +b(ba) + bb^2. 
  \end{array}
  \label{eq6}
\end{equation}

Finally, we have generated $2684$ APN functions, $336$ having perfect avalanche,
and bit independence scores, i.e. $A_{APN_i}=0$ and $B_{APN_i}=0$.

For instance, one of the APN function obtained is given in Table~\ref{tab:apncube}.

\begin{table}[htb]\center
\begin{tabular}{|c||c|c|c|c|c|c|c|c|c|c|c|c|c|c|c|c|}
  \hline
  & 0 & 1 &2 &3 &4 &5&6&7&8&9&a&b&c&d&e&f \\
  \hline
  \hline
  0&00 & 01 & 04 & 0f & 0f & 08 & 02 & 0f & 02 & 04 & 08 & 04 & 01 & 01 & 02 & 08\\ \hline 
  1&cf & fa & c4 & fb & 12 & 21 & 10 & 29 & 7c & 4e & 79 & 41 & ad & 99 & a1 & 9f\\ \hline 
  2&38 & 58 & 8e & e4 & 93 & f5 & 2c & 40 & 32 & 55 & 8a & e7 & 95 & f4 & 24 & 4f\\ \hline 
  3&a4 & f0 & 1d & 43 & dd & 8f & 6d & 35 & 1f & 4c & a8 & f1 & 6a & 3f & d4 & 8b\\ \hline 
  4&a4 & 6a & f1 & 35 & 43 & 8b & 1f & dd & d4 & 1d & 8f & 4c & 3f & f0 & 6d & a8\\ \hline 
  5&e2 & 18 & b8 & 48 & d7 & 2b & 84 & 72 & 23 & de & 77 & 80 & 1a & e1 & 47 & b6\\ \hline 
  6&b1 & 1e & 56 & f3 & f2 & 5b & 1c & bf & c9 & 61 & 20 & 82 & 86 & 28 & 66 & c2\\ \hline 
  7&a4 & 3f & 4c & dd & 35 & a8 & d4 & 43 & 6d & f1 & 8b & 1d & f0 & 6a & 1f & 8f\\ \hline 
  8&b1 & 28 & 61 & f2 & bf & 20 & 66 & f3 & 1c & 82 & c2 & 56 & 1e & 86 & c9 & 5b\\ \hline 
  9&38 & 95 & e7 & 40 & e4 & 4f & 32 & 93 & 24 & 8e & f5 & 55 & f4 & 58 & 2c & 8a\\ \hline 
  a&e2 & 1a & 80 & 72 & 48 & b6 & 23 & d7 & 47 & b8 & 2b & de & e1 & 18 & 84 & 77\\ \hline 
  b&38 & f4 & 55 & 93 & 40 & 8a & 24 & e4 & 2c & e7 & 4f & 8e & 58 & 95 & 32 & f5\\ \hline 
  c&cf & 99 & 4e & 12 & 29 & 79 & a1 & fb & 10 & 41 & 9f & c4 & fa & ad & 7c & 21\\ \hline 
  d&cf & ad & 41 & 29 & fb & 9f & 7c & 12 & a1 & c4 & 21 & 4e & 99 & fa & 10 & 79\\ \hline 
  e&b1 & 86 & 82 & bf & f3 & c2 & c9 & f2 & 66 & 56 & 5b & 61 & 28 & 1e & 1c & 20\\ \hline 
  f&e2 & e1 & de & d7 & 72 & 77 & 47 & 48 & 84 & 80 & b6 & b8 & 18 & 1a & 23 & 2b\\ \hline

\end{tabular}
\caption{An APN function generated via a pseudo-cube function over a product of semifields}\label{tab:apncube}
\end{table}

\subsection[S-Boxes based on degree 4 extensions of S(4)]{\SBoxes based on $\mathbb{S}_4^4$}
In~\cite{Combarro:2011:advances_sf}, the classification of semifields of order $256$ has been
done for characteristic four, i.e. $\mathbb{S}_{4^4}$. 
We thus have also tried to construct \SBoxes based on all these $28$
semifields up to isotopy, by using the inverse function. However, none of the
thus generated \SBoxes had a couple $(\delta, \lambda) = (4,16)$,  apart from
the one build in the semifield isomorphic to $\F_{2^8}$.

\section{Conclusion}\label{sec:concl}

In order to construct new efficient $8\times 8$ bijective \SBoxes, we replace the usual finite fields algebraic structure by semifields. However, our current knowledge about this subject does not allow us to construct directly $\mathbb{S}_{2^8}$.
We therefore build pseudo-extensions of degree $2$ of $\mathbb{S}_{2^4}$. Pseudo-extensions are based on the notion of pseudo-irreducibility, derived from a characterisation of polynomial irreducibility in finite fields. 
This allows us to define in the product of semifields, a novel function as close
as possible to the inverse function in a finite field. We call it a
pseudo-inverse and use it to build new \SBoxes. 
Many of the obtained \SBoxes have then very good evaluations on different
criterion for cryptographic resitance. Indeed, we obtained $120$ \SBoxes with
better scores than those of already known \SBoxes, including AES and
Camellia. 
We then used the same technique to generate $2684$ novel APN functions by mimicking the cube function. 

About bijective \SBoxes and APN functions, some more exhaustive search could be
done via associative variations
of Equations~(\ref{eq2}) and~(\ref{eq6}).
It could also be interesting to try to adapt to semifields other functions 
(bijective or not), like the ones described in~\cite[\S6]{Alvarez:2008:apn}. 
\bibliographystyle{plainurl}
\bibliography{semifieldsbox}

\end{document}